\documentclass[aps,prl,reprint,groupedaddress]{revtex4-1}

\usepackage{dcolumn}
\usepackage{bm}
\usepackage{braket}
\usepackage{amsfonts}
\usepackage{multirow}
\usepackage{amsmath}
\usepackage{arydshln}
\usepackage{color}
\usepackage{amsthm,paralist}
\usepackage{mathtools}
\usepackage{amssymb}

\usepackage{times}
\usepackage{tikz}
\usepackage[hidelinks,colorlinks=true,urlcolor=blue,linkcolor=blue,citecolor=red]{hyperref}

\newcommand{\F}{\mathbb{F}}
\newtheorem{theorem}{Theorem}
\newtheorem{corollary}{Corollary}
\newtheorem{lemma}{Lemma}

\begin{document}

\title{Communication Efficient Quantum Secret Sharing}

\author{Kaushik Senthoor}
\author{Pradeep Kiran Sarvepalli}
\affiliation{Department of Electrical Engineering, Indian Institute of Technology Madras, Chennai 600 036, India}

\date{May 16, 2019}

\begin{abstract}
In the standard model of quantum secret sharing, typically, one is interested in minimal authorized sets for the reconstruction of the secret. In such a setting, reconstruction requires the communication of all the shares of the corresponding authorized set. If we allow for non-minimal authorized sets, then we can trade off the size of the authorized sets with the amount of communication required for reconstruction. Based on the staircase codes, proposed by Bitar and El Rouayheb, we propose a class of quantum threshold secret sharing schemes that are also communication efficient. We call them $((k,2k-1,d))$ communication efficient quantum secret sharing schemes where $k\leq d\leq2k-1$. Using the proposed construction, we can recover a secret of $d-k+1$ qudits by communicating $d$ qudits whereas using the standard $((k,2k-1))$ quantum secret sharing requires $k(d-k+1)$ qudits to be communicated. In other words, to share a secret of one qudit, the standard quantum secret sharing requires $k$ qudits whereas the proposed schemes communicate only $\frac{d}{d-k+1}$ qudits per qudit in the communication complexity. Proposed schemes can reduce communication overheads by a factor $O(k)$ with respect to standard schemes, when $d$ equals $2k-1$. Further, we show that our schemes have optimal communication cost for secret reconstruction. 
\end{abstract}


\maketitle

\noindent
{\em Introduction. } 
A quantum secret sharing (QSS) scheme is a protocol by which a dealer can distribute an arbitrary secret state 
(in an encoded form) among $n$ participants so that only authorized subsets of participants can reconstruct the secret \cite{hillery99,karlsson99,cleve99,gottesman00,smith99,ogawa05,markham08,zhang15}. 
The secret can be a classical or  quantum state. 
The states distributed to the participants are called shares. 
Following the distribution of the secret by the dealer, certain subsets of the participants can, at a later time, recover the secret. 

A subset of parties that can reconstruct the secret is called an authorized set. 
Any subset of parties that have no information about the secret is called an unauthorized set. 
In this paper we are only interested in perfect secret sharing schemes where a subset is either authorized or unauthorized. 
In reconstruction phase, the participants constituting an authorized set pool their shares together and then recover the secret. 
Alternatively, the participants could communicate their shares to a third party or user, called the combiner, whose job is to recover the 
secret from the data communicated to the combiner. 
In this model, a metric of interest is the amount of communication between the participants and the combiner. 
The amount of communication  from the participants to the combiner is called the communication cost.  

In this paper, we initiate the study of communication efficient quantum secret sharing schemes for quantum secrets,  
opening a new avenue for further research in quantum secret sharing. 
We propose schemes which aim to minimize the communication cost of quantum secret sharing schemes. 
While the problem of communication cost in classical secret sharing schemes was studied previously, \cite{bitar16,huang16,huang17,wang08,penas18,yan18}, the corresponding problem for quantum secret sharing schemes has not been studied thus far. 
Quantum secret sharing has become experimentally viable and there are many demonstrations, 
see for instance \cite{tittel01,imai03,wei13,hao11,bogdanski08,bell14,schmid05,gaertner07}. 
However, quantum information is still an expensive resource, and clearly, we would like to reduce the 
cost of storing and transmitting  it. 
Our results should be of interest to experimentalists  as well. 

The collection of authorized sets is called the access structure (denoted as $\Gamma$) of the secret sharing scheme. 
We focus on an important class of secret sharing schemes, namely, the $((k,n))$ quantum  threshold schemes (QTS)
where any subset of $t$  participants with  $k\leq t\leq n$ can reconstruct the secret.

\noindent
{\em Contributions.} 
Based on the staircase codes proposed by Bitar {\em et al.} \cite{bitar16}, we propose a class of  quantum  threshold secret sharing schemes that are also communication efficient. In the standard model of quantum secret sharing, sharing a secret of one qudit using a $((k,2k-1))$ threshold scheme requires $k$ qudits to be communicated to reconstruct the secret. In the proposed schemes, we can recover the secret of $m=d-k+1$ qudits by communicating  $d$ qudits where $k<d\leq 2k-1$, in average $\frac{d}{d-k+1}$ qudits for every qudit in secret.  
Further, we show that these schemes are optimal with respect to communication cost in the given model of quantum secret sharing.  

\noindent
{\em Previous Work.}
The closest work related to ours appears to be that of \cite{ben12} who also aimed at reducing the communication cost in 
quantum secret sharing schemes. 
However, there are important differences, their work uses a combination of non-perfect secret sharing schemes along with 
a hybrid quantum secret sharing scheme. A hybrid QSS scheme is one which participants have  (partly or wholly) classical shares. 
Our schemes in contrast are purely quantum in that no share is classical. 
Furthermore, the work in \cite{ben12} is concerned with the communication cost of the secret sharing schemes during distribution of the (encoded) secret more than the cost during reconstruction  which is our focus here. 

\noindent
{\em A Motivating Example.}
The intuition behind the communication efficient secret sharing schemes lies in using a nonminimal authorized set to recover the secret. 
(An authorized set is said to be a minimal authorized set if every proper subset of the authorized set is unable to recover the secret.)
Let $\F_q$ denote the finite field with $q$ elements.
Consider the ternary $((2,3))$ quantum threshold scheme proposed by Cleve {\em et al} \cite{cleve99}.
In this scheme, the secret state $s \in \F_3$ is encoded into three qudits 
as 
$\ket{s}\mapsto\frac{1}{\sqrt{3}}\sum_{r=0}^{2}\ket{r}_A\ket{s+r}_B\ket{2s+r}_C$ where the one qudit each is given to parties $A$, $B$ and $C$.
In order to reconstruct the secret we need to communicate two qudits to the combiner. 
We propose an alternate $((2,3))$ quantum threshold scheme where we can obtain better communication costs. 
In this scheme $(s_1,s_2)\in \F_3^2$ is encoded as follows:
\begin{eqnarray}
\ket{s_1 s_2}\ \mapsto\   \sum_{r_1, r_2\in \F_3}
\begin{array}{c}
\ket{s_1+r_1,r_2}_A\\
\ket{s_2+r_1,r_1+r_2}_B\\
\ket{s_1+s_2+r_1,r_1+2r_2}_C
\end{array} \label{eq:23ceqss-1}
\end{eqnarray}
where we have ignored the normalizing factors. 
In this case, the secret is encoded into six qudits, equivalently each secret qudit is encoded into three qudits as in previous scheme \cite{cleve99}. 

Let us look at the reconstruction of the secret 
from the four qudits of the first two participants $A$ and $B$ from the state as given in Eq.~\eqref{eq:23ceqss-1}. 
The reconstruction steps are similar for other choices of two participants as well.
(Values of qudits which have changed after each operation are indicated in bold.)
By subtracting the value of second qudit of $A$ from that of $B$, we can obtain the following  state.
\begin{eqnarray}
\sum_{r_1,r_2\in \F_3}
\begin{array}{c}{\ket{s_1+r_1,r_2}}_A\\
{\ket{s_2+r_1,\bm{r_1}}}_B\\
\ket{s_1+s_2+r_1,r_1+2r_2}_C
\end{array}\nonumber  
\end{eqnarray}
By subtracting the value of the second qudit of $B$ from the values of the first qudits of $A$ and $B$, 
\begin{eqnarray}
\sum_{r_1,r_2\in \F_3}  
\begin{array}{c}\ket{\bm{s_1},r_2}_A \ket{\bm{s_2},r_1}_B \ket{s_1+s_2+r_1,r_1+2r_2}_C
\\
\end{array}  \nonumber
\end{eqnarray}
We can now obtain the following state
\begin{eqnarray}
\ket{s_1}_A\ket{s_2}_B\sum_{r_1,r_2\in \F_3} 
\begin{array}{c}
\ket{\bm{s_1+s_2+r_1}}_B\ket{\bm{r_1+2r_2}}_A\\\ket{s_1+s_2+r_1}_C\ket{r_1+2r_2}_C
\end{array} \nonumber
\end{eqnarray}
This does not end the reconstruction process because the secret could be still entangled with the rest of the system
and we may not be able to recover an arbitrary superposition. 
Further steps are required to recover the secret completely. 
Setting $t_1= s_1+s_2+r_1$ and then $t_2=t_1+2s_1+2s_2+2r_2$, we obtain the following state
\begin{eqnarray}
\ket{s_1}_A\ket{s_2}_B\sum_{t_1,r_2\in\F_3} \ket{t_1}_B\ket{t_1+2s_1+2s_2+2r_2}_A \nonumber\\
[-0.4cm]\ket{t_1}_C\ket{t_1+2s_1+2s_2+2r_2}_C\hspace{-0.5cm} \nonumber\\
=\ket{s_1}_A\ket{s_2}_B\sum_{t_1,t_2\in\F_3} \ket{t_1}_B\ket{t_2}_A\ket{t_1}_C\ket{t_2}_C \nonumber
\end{eqnarray}
At this point the secret is found to be completely disentangled with the rest of the qudits and 
the state of the remaining qudits is independent of the secret, thereby ensuring we can recover an arbitrary linear combination of basis states.

Let us recover the secret when we have access to all three participants (they constitute a non-minimal authorized set).
We do not need to have  access to all the six qudits of the participants. 
Just three qudits i.e., one qudit, specifically the first qudit,  from each share will suffice. 
By unitary operations on these three qudits of the state in Eq.~\eqref{eq:23ceqss-1}, we  obtain,
\begin{eqnarray}
\sum_{r_1,r_2\in \F_3} \ket{\bm{s_1},r_2}_A\ket{\bm{s_2},r_1+r_2}_B\ket{\bm{r_1},r_1+2r_2}_C \nonumber
\end{eqnarray}
Reordering the qudits, we have 
\begin{eqnarray*}
\ket{s_1}_A\ket{s_2}_B\sum_{r_1,r_2} \ket{r_1}_C\ket{r_2}_A\ket{r_1+r_2}_B\ket{r_1+2r_2}_C
\end{eqnarray*}
Once again the secret is completely disentangled from the rest of the system and we are able to recover the secret using only 
three qudits. 
However, note that in this case we are able to recover a secret of two qudits. Had we used the $((2,3))$ threshold scheme of 
Cleve {\em et al.}, we would have needed four qudits even when we allow access to all the three participants. 
This example demonstrates we can reduce the number of qudits to be communicated  when reconstructing the secret.

\smallskip
\noindent
{\em Proposed Communication Efficient Quantum Secret Sharing Schemes.}
To specify a quantum secret sharing concretely, we give the encoding for the basis states of the secret. 
An encoding $\mathcal{E}$ realizes a perfect quantum secret sharing scheme with access structure $\Gamma$ if it satisfies the following 
constraints \cite{imai03}.

\begin{compactenum}[i)]
\item (Recoverability) Any  set in $\Gamma$ can recover the secret. 
\item (Secrecy) Any set not in $\Gamma$  has no information about the secret.
\end{compactenum} 
To show recoverability, we explictly show that the set can recover the secret. 
To show secrecy, we show that the complement of the set contains an authorized set. 
A quantum secret sharing scheme is said to be a pure state scheme if encodes pure state secrets to global pure states.

We denote by $((k,n,d))_q$ a $q$-ary quantum threshold scheme with $n$ participants, where any $k$ participants can recover the secret and 
$d>k$ participants can recover the secret with lower communication cost. 
We suppress the subscript for convenience. 
We assume that number of participants is $n=2k-1$ and fewer than  $k$ cannot recover the secret. 
Fix an integer $k\leq d\leq n$,  and a prime $q>n$. The secret contains $m$ qudits where each qudit is $q$-dimensional and 
\begin{eqnarray}
m= d-k+1. \label{eq:secretSize}
\end{eqnarray}

Consider the vectors $\underline{s}$$\ =$$\ (s_1,s_2, \hdots, s_m)$ in $\ \mathbb{F}_q^m$ and $\underline{r}$$\ =$$\ (r_1,r_2, \hdots, r_{m(k-1)})$ in $\ \mathbb{F}_q^{m(k-1)}$. The vector $\underline{r}$ is further split into $m$ vectors $\underline{r}_1=(r_1, r_2, \hdots, r_{k-1})$, $\underline{r}_2=(r_k,r_{k+1}, \hdots, r_{2(k-1)}),\  \hdots$ $\underline{r}_m=(r_{(m-1)(k-1)+1},r_{(m-1)(k-1)+1}, \hdots, r_{m(k-1)})$. The vector $\underline{r}_1$ alone is further split into two vectors with its first $(k-m)$ values in $\underline{u}$ and the remaining $(m-1)$ values in $\underline{v}$.

Let $x_1, x_2,..., x_n$ be distinct nonzero elements from $\F_q$.
Define the Vandermonde matrix $V_{n,d}$ to be the $n\times d$ matrix whose $(i,j)$th entry 
is given by $x_i^{j-1}$ for $1\leq i\leq n$ and $1\leq j\leq d$.
Assume that $V_{n,d}$ is known to all the parties involved.
Let $s_i,r_j\in \F_q$, where $1\leq i\leq m$ and $1\leq j\leq m(k-1)$. We define the following matrix $Y$. \\\ \\
\begin{tikzpicture}
\hspace{0.3cm}
\draw [dashed] (0,0) -- (0,0);
\draw [dashed] (1.6,0) -- (1.6,3.6);
\end{tikzpicture}
\vspace{-4cm}{\small
\begin{eqnarray}
\left[
    \begin{array}{cc}
        \begin{matrix} s_1 \\ s_2 \\ \vdots \end{matrix} & \text{\huge 0}_{(m-1)\times(m-1)} \\
        \cdashline{2-2}[4pt/4pt]
        s_m & \begin{matrix} \hspace{-0.25in}r_{k-m+1} & r_{k-m+2} & \ \hdots & \ \ \ \ \ r_{k-1} \end{matrix} \\
        \cdashline{1-2}[4pt/4pt]
        \begin{matrix} r_1\\ r_2\\ \vdots\\ r_{k-1} \end{matrix} & \begin{matrix} \ r_k            & r_{2(k-1)+1} &  \hdots & r_{(m-1)(k-1)+1} \\
                                                                                                                      \ r_{k+1}      & r_{2(k-1)+2} &  \hdots & r_{(m-1)(k-1)+2} \\
                                                                                                                      \ \vdots        & \vdots          &  \ddots & \vdots \\
                                                                                                                      \ r_{2(k-1)} &  r_{3(k-1)}   &  \hdots & r_{m(k-1)}\end{matrix}
\end{array}
\right]\label{eq:msgMatrix}
\end{eqnarray}
}

\vspace{-0.4cm}
We also represent $Y$ in a slightly compact form as follows. 
\begin{eqnarray} 
Y =
\left[
\begin{tabular}{cc}
\multirow{3}{*}{$\underline{s}\ $} & \multirow{2}{*}{\large 0}\\
&\\ \cdashline{2-2}
&$\underline{v}^t$\\
\cdashline{1-2}
\multirow{2}{*}{$\underline{r}_1\ $} & \multirow{2}{*}{$\ \underline{r}_2\ \ \underline{r}_3\ \hdots\ \underline{r}_m$}\\
&\\
\end{tabular}
\right]\label{eq:msgMat-2}
\end{eqnarray}

\vspace{-2.5cm}
\begin{tikzpicture}
\draw [dashed] (0,0) -- (0,0);
\draw [dashed] (3.5,0) -- (3.5,2.1);
\end{tikzpicture}

\vspace{0.2cm}
Consider the matrix $C=V_{n,d}Y$  where $Y$ is defined as in Eq.~\eqref{eq:msgMatrix}. 
Each entry in matrix $C$, $c_{ij}$ is a function of $\underline{s}$ and $\underline{r}$. The encoding for the basis states $(s_1,\ldots, s_m)\in \F_q^m$ is given by 
$\mathcal{E}$, where 
\begin{eqnarray}
\mathcal{E}:\ket{s_1 s_2\hdots s_m}\ \mapsto\sum_{\underline{r}\in\mathbb{F}_q^{m(k-1)}}
\bigotimes_{i=1}^{2k-1}\ket{c_{i1} c_{i2}\ldots c_{im} },
\label{eq:enc_qudits}
\end{eqnarray}
where we have omitted the normalizing factor. 
The qudits in the share of the $i$th participant are indexed by $i$.
The first share contains the first $m$ qudits, the second share contains the next set of $m$ qudits and so on till the $(2k-1)$th share.

\begin{lemma}[Recoverability for non-minimal authorized sets]
\label{lm:qts-d-recovery}
For the encoding scheme given in Eq.~\eqref{eq:enc_qudits}, we can recover the secret from  any 
$d$ shares by accessing only the first qudit in each share.
\end{lemma}
\begin{proof}
We shall prove this by giving the sequence of operations to be performed so that the $d$ shares can recover the secret
with {\em only} $d$ qudits. 
Each of the $d$ participants sends their first qudit to the combiner for reconstructing the secret. 
Let $D = \{i_1, i_2, \hdots, i_d\} \subset \{1,2,\hdots,2k-1\}$ be the set of $d$ shares chosen and $E=\{i_{d+1},i_{d+2},\hdots,i_{2k-1}\}$ be the complement of $D$. Let $V_D$ and $V_E$ be the matrices containing the rows of $V_{n,d}$ corresponding to $D$ and $E$ respectively. 
Then, Eq.~\eqref{eq:enc_qudits} can be rearranged as
\begin{eqnarray*}
\sum_{\underline{r}\in\mathbb{F}_q^{m(k-1)}}
\textcolor{purple}{\ket{c_{i_1,1}\ c_{i_2,1}...c_{i_d,1}}}
\ket{c_{i_{d+1},1}\ c_{i_{d+2},1}...c_{i_{2k-1},1}}\hspace{5cm}\\
[-0.4cm]\ket{(c_{i_1,2}\ c_{i_2,2}...c_{i_{2k-1},2}) 
...(c_{i_1,m}\ c_{i_2,m}...c_{i_{2k-1},m})},\hspace{4cm}
\end{eqnarray*}
where we have highlighted (in color) the qudits accessed by the combiner.
Now using the fact that $c_{ij}$ is the product of $i$th row of $V_{n,q}$ and $j$th column of $Y$
and $\underline{r} = (\underline{r}_1, \ldots, \underline{r}_m)$, we can rewrite this as
\begin{eqnarray*}
\sum_{\underline{r}\in\mathbb{F}_q^{m(k-1)}}
\textcolor{purple}{\ket{V_D(\underline{s},\underline{r}_1)}}\ket{V_E(\underline{s},\underline{r}_1)}
\ket{V(\underline{0},r_{k-m+1},\underline{r}_2)}\cdots \hspace{4.5cm}
\\[-0.35cm]\cdots\ 
 \ket{V(\underline{0},r_{k-1},\underline{r}_m)}
\hspace{4cm}
\end{eqnarray*}
Since $V_D$ is a $d\times d $ Vandemonde matrix of full rank, we can apply $V_D^{-1}$ to the $d$ qudits 
with the combiner to transform the state as follows. 
\begin{eqnarray*}
\sum_{\underline{r}\in\mathbb{F}_q^{m(k-1)}}
\textcolor{purple}{\ket{\underline{s},\underline{r}_1}}\ket{V_E(\underline{s},\underline{r}_1)}
\ket{V(\underline{0},r_{k-m+1},\underline{r}_2)}\cdots
\hspace{4.5cm}
\\[-0.35cm]\cdots\ 
\ket{V(\underline{0},r_{k-1},\underline{r}_m)}
\hspace{4cm}
\end{eqnarray*}
Then from Eq.~\eqref{eq:msgMatrix} we have $\underline{r}_1 = (\underline{u}, \underline{v})$, 
and $r_{k-m+j} = v_j$ for $1\leq j\leq m-1$, we can write 
\begin{eqnarray*}
\textcolor{purple}{\ket{\underline{s}}}
\sum_{\substack{(\underline{v},\underline{r}_2,\underline{r}_3,..\underline{r}_m)\\\in\mathbb{F}_q^{k(m-1)}}}
\sum_{\underline{u}\in\mathbb{F}_q^{k-m}}
\textcolor{purple}{\ket{\underline{u}}\ket{\underline{v}}}\ket{V_E(\underline{s},\underline{u},\underline{v})}
\ket{V(\underline{0},v_1,\underline{r}_2)}\cdots\hspace{4cm} 
\\[-0.4cm]\cdots\ket{V(\underline{0},v_{m-1},\underline{r}_m)}
\hspace{4cm}
\end{eqnarray*}
Since the combiner has access to $\ket{\underline{s}}$, $\ket{\underline{u}}$, and $\ket{\underline{v}}$, 
we can use the matrix $V_E$, of rank $k-m$ equal to the size of $\underline{u}$, to transform $\ket{\underline{u}}$ to $\ket{V_E(\underline{s},\underline{u},\underline{v})}$.
\begin{eqnarray*}
\textcolor{purple}{\ket{\underline{s}}}\sum_{\substack{(\underline{v},\underline{r}_2,\underline{r}_3,..\underline{r}_m)\\\in\mathbb{F}_q^{k(m-1)}}}
\sum_{\underline{u}\in\mathbb{F}_q^{k-m}}
\textcolor{purple}{\ket{V_E(\underline{s},\underline{u},\underline{v})}\ket{\underline{v}}}\ket{V_E(\underline{s},\underline{u},\underline{v})}\hspace{10cm}
\\[-0.7cm]\hspace{-1cm}\ket{V(\underline{0},v_1,\underline{r}_2)}\ \cdots \ket{V(\underline{0},v_{m-1},\underline{r}_m)}
\hspace{8.8cm}
\end{eqnarray*}
Rearranging qudits $\textcolor{purple}{\ket{\underline{v}}}\ket{V_E(\underline{s},\underline{u},\underline{v})}$ to $\ket{V_E(\underline{s},\underline{u},\underline{v})}\textcolor{purple}{\ket{\underline{v}}}$,\vspace{-0.2cm}
\begin{eqnarray*}
\textcolor{purple}{\ket{\underline{s}}}\sum_{\substack{(\underline{v},\underline{r}_2,\underline{r}_3,..\underline{r}_m)\\\in\mathbb{F}_q^{k(m-1)}}}
\Bigg(\sum_{\underline{u}\in\mathbb{F}_q^{k-m}}
\textcolor{purple}{\ket{V_E(\underline{s},\underline{u},\underline{v})}}\ \ket{V_E(\underline{s},\underline{u},\underline{v})}\Bigg)\ \textcolor{purple}{\ket{\underline{v}}}
\\[-0.4cm]\ket{V(\underline{0},v_1,\underline{r}_2)} \cdots \ket{V(\underline{0},v_{m-1},\underline{r}_m)}
\hspace{-0.5cm}
\end{eqnarray*}
Since $E$ is of size  $(2k-1-d)$, with Eq.~\eqref{eq:secretSize}, we see that $V_E$ is a Vandermonde matrix of size 
 $(k-m)\times d$ and rank $k-m <d$. Therefore, the image of $V_E$ spans $\F_q^{k-m}$ and 
$ \sum_{\underline{u}\in \F_q^{k-m}} \ket{V_E(\underline{s},\underline{u},\underline{v})}\ket{V_E(\underline{s},\underline{u},\underline{v})}$ is independent of $\underline{s}$. The state can be written as
\begin{eqnarray*}
\hspace{0cm}
\textcolor{purple}{\ket{\underline{s}}}
\sum_{\underline{f}\in\mathbb{F}_q^{k-m}}
\textcolor{purple}{\ket{\underline{f}}}\ket{\underline{f}}
\sum_{\substack{(\underline{v},\underline{r}_2,\underline{r}_3,..\underline{r}_m)\\\in\mathbb{F}_q^{k(m-1)}}}
\textcolor{purple}{\ket{\underline{v}}}
\ket{V(\underline{0},v_1,\underline{r}_2)}\hspace{1.3cm}
\\[-0.7cm]\hdots\ket{V(\underline{0},v_{m-1},\underline{r}_m)}
\hspace{0cm}
\end{eqnarray*}
The secret is now completely disentangled from the rest of the system, therefore even when the secret is an arbitrary superposition we can recover the secret from $d$ shares as claimed. 
\end{proof}

\begin{lemma}[Recoverability for minimal authorized sets]
\label{lm:qts-k-recovery}
For the encoding scheme given in Eq.~\eqref{eq:enc_qudits}, we can recover the secret by accessing (all) the qudits of any
$k$ shares.
\end{lemma}
\begin{proof}
For secret recovery from $k$ shares, all the qudits from each chosen share are sent to the user. 
Let $K=\{j_1,j_2,\hdots,j_k\}\subset \{1,2,\hdots,2k-1\}$ be the set of $k$ shares chosen and $L=\{j_{k+1},j_{k+2},\hdots,j_{2k-1}\}$ be the complement of $K$. Let $V_K$ and $V_L$ be the matrices containing the rows of $V_{n,d}$ corresponding to $K$ and $L$ respectively. 
 Then, grouping the ($i$th) qudits of $K$ and $L$, the encoded state  in Eq.~\eqref{eq:enc_qudits} can be written as 
\begin{eqnarray*}
\sum_{\underline{r}\in\mathbb{F}_q^{m(k-1)}}
\textcolor{purple}{\ket{c_{j_1,1}\ c_{j_2,1}\hdots c_{j_k,1}} \cdots 
\ket{c_{j_1,m}\ c_{j_2,m}\hdots c_{j_k,m}}}
\\[-0.4cm]\hspace{0cm}
\ket{c_{j_{k+1},1}\ c_{j_{k+2},1}\hdots c_{j_{2k-1},1}}\hspace{2cm}\\
\cdots\ket{c_{j_{k+1},m}\ c_{j_{k+2},m}\cdots c_{j_{2k-1},m}}
\end{eqnarray*}
This can be written in terms of $V_K$ and $V_L$ as 
\begin{eqnarray*}
\sum_{\underline{r}\in\mathbb{F}_q^{m(k-1)}}
\textcolor{purple}{\ket{V_K(\underline{s},\underline{r}_1)} \ket{V_K(\underline{0},r_{k-m+1},\underline{r}_2)}\cdots}
\hspace{7.5cm}
\\[-0.35cm]\hspace{0cm}
\textcolor{purple}{\ket{V_K(\underline{0},r_{k-1},\underline{r}_m)}}
\hspace{7cm}
\\\hspace{0cm}
\ket{V_L(\underline{s},\underline{r}_1)}\ket{V_L(\underline{0},r_{k-m+1},\underline{r}_2)}\cdots
\hspace{7cm}
\\\hspace{0cm}
\ket{ V_L(\underline{0},r_{k-1},\underline{r}_m)}
\hspace{6.5cm}
\end{eqnarray*}
Letting $V_{K,\ell }$ be the submatrix of $V_K$ consisting of the last $k$ columns. 
Then we can simplify the state as
\begin{eqnarray*}
\sum_{\underline{r}\in\mathbb{F}_q^{m(k-1)}}
\hspace{-0.2cm}\textcolor{purple}{\ket{V_K(\underline{s},\underline{r}_1)} \ket{V_{K,\ell}(r_{k-m+1},\underline{r}_2)}...\ket{V_{K,\ell}(r_{k-1},\underline{r}_m)}}
\hspace{6.45cm}
\\[-0.4cm]
\ket{V_L(\underline{s},\underline{r}_1)} \ket{V_L(\underline{0},r_{k-m+1},\underline{r}_2)}...
\ket{V_L(\underline{0},r_{k-1},\underline{r}_m)}
\hspace{6.2cm}
\end{eqnarray*}
Since $V_{K,\ell}$ is a $k\times k $ Vandermonde matrix of full rank, we can apply ${V_{K,\ell}}^{-1}$ to further transform the state as 
\begin{eqnarray*}
\sum_{\underline{r}\in\mathbb{F}_q^{m(k-1)}}
\textcolor{purple}{\ket{V_K(\underline{s},\underline{r}_1)}
\ket{r_{k-m+1},\underline{r}_2}...\ket{r_{k-1},\underline{r}_m}}\ket{V_L(\underline{s},\underline{r}_1)} 
\hspace{7cm}
\\[-0.3cm]\ \ \ \ \ket{V_L(\underline{0},r_{k-m+1},\underline{r}_2)}...\ket{V_L(\underline{0},r_{k-1},\underline{r}_m)}
\hspace{8cm}
\end{eqnarray*}

Then from Eq.~\eqref{eq:msgMatrix} we have $\underline{r}_1 = (\underline{u}, \underline{v})$, 
and $r_{k-m+j} = v_j$ is the $j$th entry in $\underline{v}$ for $1\leq j\leq m-1$, and rearranging the qudits, we can write the state as  
\begin{eqnarray*}
\sum_{\underline{r}\in\mathbb{F}_q^{m(k-1)}}
\textcolor{purple}{\ket{V_K(\underline{s},\underline{u},\underline{v})}\ket{\underline{v}}\ket{\underline{r}_2,\underline{r}_3,\hdots\underline{r}_m}}
\ \ \ket{V_L(\underline{s},\underline{r}_1)}\hspace{4.5cm}\\
[-0.3cm]\ket{V_L(\underline{0},r_{k-m+1},\underline{r}_2)}\cdots \ket{V_L(\underline{0},r_{k-1},\underline{r}_m)}\hspace{4.5cm}
\end{eqnarray*}
Let $V_{K,f}$ be the first $k$ columns of $V_K$ and $V_{K,\bar{f}} $ be the submatrix of remaining columns.
Note that $V_{K,\bar{f}}$ has $m-1$ columns. 
Then  $V_K(\underline{s},\underline{u},\underline{v}) = V_{K,f}(\underline{s},\underline{u})+V_{K,\bar{f}}(\underline{v})$. Thus, the above state can be written as,
\begin{eqnarray*}
\sum_{\underline{r}\in\mathbb{F}_q^{m(k-1)}} 
\textcolor{purple}{\ket{V_{K,f}(\underline{s},\underline{u})+V_{K,\bar{f}}(\underline{v})}\ket{\underline{v}}\ket{\underline{r}_2,\underline{r}_3,\hdots\underline{r}_m}}\hspace{6cm}\\
[-0.15in]\ket{V_L(\underline{s},\underline{r}_1)}\hspace{9.5cm}\\
\ket{V_L(\underline{0},r_{k-m+1},\underline{r}_2)}\cdots \ket{V_L(\underline{0},r_{k-1},\underline{r}_m)}
\hspace{5cm}
\end{eqnarray*} 
At this point the combiner has access to $\ket{\underline{v}}$ and can subtract $V_{K,\bar{f}}(\underline{v})$ from $\ket{V_{K,f}(\underline{s},\underline{u})+V_{K,\bar{f}}(\underline{v})}$  to obtain 
\begin{eqnarray*}
\sum_{\underline{r}\in\mathbb{F}_q^{m(k-1)}}
\textcolor{purple}{\ket{V_{K,f}(\underline{s},\underline{u})}\ket{\underline{v}}\ket{\underline{r}_2,\underline{r}_3,\hdots\underline{r}_m}}
\ \ \ket{V_L(\underline{s},\underline{u},\underline{v})}\hspace{4.5cm}\\
[-0.1in]\ket{V_L(\underline{0},r_{k-m+1},\underline{r}_2)}\cdots \ket{ V_L(\underline{0},r_{k-1},\underline{r}_m)}
\hspace{5cm}
\end{eqnarray*} 
Since $V_{K,f}$ is a $k\times k $ Vandermonde matrix of full rank, we can apply ${V_{K,f}}^{-1}$ to extract $\ket{s}$ as shown below.
\begin{eqnarray*}
\sum_{\underline{r}\in\mathbb{F}_q^{m(k-1)}}
\textcolor{purple}{\ket{\underline{s}}\ket{\underline{u}}\ket{\underline{v}}\ket{\underline{r}_2,\underline{r}_3,\hdots\underline{r}_m}}
\ \ \ket{V_L(\underline{s},\underline{u},\underline{v})}\hspace{7cm}\\
[-0.1in]\ket{V_L(\underline{0},r_{k-m+1},\underline{r}_2)} \cdots \ket{V_L(\underline{0},r_{k-1},\underline{r}_m)}
\hspace{6cm}\\
[0.15in]=\textcolor{purple}{\ket{\underline{s}}}\sum_{\underline{r}\in\mathbb{F}_q^{m(k-1)}}
\textcolor{purple}{\ket{\underline{r}_1}\ket{\underline{r}_2,\underline{r}_3,\hdots\underline{r}_m}}
\ \ \ket{V_L(\underline{s},\underline{r}_1)}\hspace{7.5cm}\\
[-0.1in]\ket{V_L(\underline{0},r_{k-m+1},\underline{r}_2)}\cdots \ket{V_L(\underline{0},r_{k-1},\underline{r}_m)}
\hspace{6cm}
\end{eqnarray*}
Since $V_L$ is a $(k-1)\times d$ matrix of rank $k-1$, we can now modify each of the registers $\ket{\underline{r}_i}$ of size $(k-1)$ qudits, $\ket{\underline{r}_1}$ to $\ket{V_L(\underline{s},\underline{r}_1}$ and $\ket{\underline{r}_i}$ for $2\leq i\leq m$, to $\ket{V_L(\underline{0},r_{k-m+i-1},\underline{r}_i)}$. 
\begin{eqnarray*}
\textcolor{purple}{\ket{\underline{s}}}\sum_{\underline{r}\in\mathbb{F}_q^{m(k-1)}}
\textcolor{purple}{\ket{V_L(\underline{s},\underline{r}_1)}}\hspace{9cm}\\
[-0.5cm]\textcolor{purple}{\ket{V_L(\underline{0},r_{k-m+1},\underline{r}_2)\hdots V_L(\underline{0},r_{k-1},\underline{r}_m)}}\hspace{4.5cm}\\
\ket{V_L(\underline{s},\underline{r}_1)}\hspace{8.3cm}
\\\ket{V_L(\underline{0},r_{k-m+1},\underline{r}_2)\hdots V_L(\underline{0},r_{k-1},\underline{r}_m)}
\hspace{3.9cm}
\end{eqnarray*}
On rearranging the qudits, we obtain 
\begin{eqnarray*}
\textcolor{purple}{\ket{\underline{s}}}\sum_{\underline{r}_1\in\mathbb{F}_q^{k-1}}\textcolor{purple}{\ket{V_L(\underline{s},\underline{r}_1)}}\ket{V_L(\underline{s},\underline{r}_1)}
\hspace{5cm}
\\[-0.2cm]
\sum_{\underline{r}_2\in\mathbb{F}_q^{k-1}}\textcolor{purple}{\ket{V_L(\underline{0},r_{k-m+1},\underline{r}_2)}}\ket{V_L(\underline{0},r_{k-m+1},\underline{r}_2)}
\hspace{1.5cm}
\\[-0.6cm]\hspace{0cm}
\ddots
\hspace{5cm}
\\\sum_{\underline{r}_m\in\mathbb{F}_q^{k-1}}\textcolor{purple}{\ket{V_L(\underline{0},r_{k-1},\underline{r}_m)}}\ket{V_L(\underline{0},r_{k-1},\underline{r}_m)}
\hspace{1.5cm}
\end{eqnarray*}
$V_L$ is a Vandermonde matrix of size $(k-1)\times d$ with $d>k-1$. 
So the image of $V_L$ is of dimension $k-1$. 
Therefore $\sum_{\underline{r}_i\in \F_q^{k-1}}\ket{V_L(\underline{0},r_{k-m+i-1}, \underline{r}_i)}\ket{V_L(\underline{0},r_{k-m+i-1}, \underline{r}_i)}$ is a uniform superposition independent of $r_{k-m+i-1}$, for $2\leq i\leq m$.
\begin{eqnarray*}
\textcolor{purple}{\ket{\underline{s}}}\sum_{\underline{r}_1\in\mathbb{F}_q^{k-1}}\textcolor{purple}{\ket{V_L(\underline{s},\underline{r}_1)}}\ket{V_L(\underline{s},\underline{r}_1)}
\hspace{5cm}
\\[-0.3cm]
\sum_{\underline{f}_2\in\mathbb{F}_q^{k-1}}\textcolor{purple}{\ket{\underline{f}_2}}\ket{\underline{f}_2}
\hdots
\sum_{\underline{f}_m\in\mathbb{F}_q^{k-1}}\textcolor{purple}{\ket{\underline{f}_m}}\ket{\underline{f}_m}
\hspace{1.3cm}
\end{eqnarray*}
  Now we can show that $\sum_{\underline{r}_1\in \F_q^{k-1}}\ket{V_L(\underline{s},\underline{r}_1)}\ket{V_L(\underline{s},\underline{r}_1)}$ is a uniform superposition
independent of $\underline{s}$, since $V_L$ has rank $k-1$.
\begin{eqnarray*}
\textcolor{purple}{\ket{\underline{s}}}
\sum_{\underline{f}_1\in\mathbb{F}_q^{k-1}}\textcolor{purple}{\ket{\underline{f}_1}}\ket{\underline{f}_1}
\cdots
\sum_{\underline{f}_m\in\mathbb{F}_q^{k-1}}\textcolor{purple}{\ket{\underline{f}_m}}\ket{\underline{f}_m}
\hspace{1.3cm}
\end{eqnarray*}
At this point the state is given by the above expression with the secret completely disentangled from the rest of the system
and we can recover any arbitrary superposition. 
This completes the proof that $k$ shares can recover the secret. 
\end{proof}

\begin{lemma}[Secrecy]
\label{lm:secrecy}
In the encoding scheme defined in Eq.~\eqref{eq:enc_qudits}, any $k-1$ or lesser number of shares do
 not give any information about the secret $\ket{s}$.
\end{lemma}

\begin{proof}
The encoding scheme is a pure state encoding scheme with the total number of shares $n=2k-1$. If some set of $k-1$ or lesser number of shares give any information about the secret, then the secret cannot be recovered  from the remaining $k$ or more number of shares, because of the no-cloning theorem. However, from Lemma \ref{lm:qts-k-recovery}, any $k$ shares are enough to recover the secret completely. Hence, no set of $k-1$ (or lesser number of) shares give any information about the secret.
\end{proof}
With these results in place we have our central result. 

\begin{theorem}[Communication efficient QSS]\label{th:ce-qss}
The encoding given in Eq.~\eqref{eq:enc_qudits} gives rise to a $((k,2k-1,d))$ quantum secret sharing 
scheme where $d$ is a fixed integer satisfying $k \leq d\leq 2k-1$. The scheme shares a secret of $m=d-k+1$ qudits.  
The communication cost for any $k$ participants to recover the secret 
is $mk$ qudits, while the communication cost for any $d$ participants is $d$ qudits. 
\end{theorem}

A standard $((k,2k-1))$ QTS will incur a communication cost of  $km$ qudits to share $m$ qudits.
A subtle point to be noted is that the communication efficient scheme requires the dealer to share a larger secret. 

An $((k,2k-1))$ QTS can be converted to $((k,n))$ QTS for $k\leq n\leq 2k-1$ by throwing away or ignoring $2k-1-n$ shares of the $((k,2k-1))$ scheme, \cite[Theorem~1]{cleve99}.  
If $n<2k-1$, then the scheme is a mixed  state scheme. 
Therefore, Theorem~\ref{th:ce-qss} implies the existence of $((k,n,d))$ quantum secret sharing schemes where $k\leq d\leq n\leq 2k-1$. 
Note that a $((k,n))$ QTS cannot exist for $n\geq 2k$ by \cite[Theorem~2]{cleve99}.

Next we show that the proposed secret sharing schemes are optimal with respect to the 
communication cost. We need the following lemma due to Gottesman \cite[Theorem~5]{gottesman00}. 

\begin{lemma}
\label{lm:lim-qc}
Even in the presence of pre-existing entanglement, sending an arbitrary state from a Hilbert space of dimension $h$ requires a channel of dimension $h$.
\end{lemma}

\begin{lemma}[Secret replacement with authorized set] \label{lm:sec-rep}
A party having access to an authorized set of shares in a quantum secret sharing scheme can replace the secret encoded with any arbitrary state (of the same dimension as the secret) without disturbing the remaining shares. After this replacement, secret recovery from any of the authorized sets will give only the new state.
\end{lemma}

\begin{proof}
Let $A\subseteq[1,n]$ be an arbitrary authorized set in the given quantum secret sharing scheme and $B$ be its complement. Let $\mathcal{E}:\mathcal{S}\rightarrow \mathcal{A}\otimes\mathcal{B}$ denote the  operation for encoding the secret and $\mathcal{R}_A:\mathcal{A}\rightarrow\mathcal{S}$ be the  operation required for recovering the secret from the authorized set $A$. 

If $\ket{\phi}$ is the secret encoded, then the encoding can be given as $\mathcal{E}\ket{\phi}\ket{0}$ where $\ket{0}$ represents the ancilla qudits.
To replace the secret $\ket{\phi}$ with the arbitrary state $\ket{\psi}$ of the same dimension, perform the following steps on the set $A$: i) Recover the secret $\ket{\phi}$ using $\mathcal{R}_A$ by acting only on $A$. The joint state with A and B becomes $\ket{\phi}\bra{\phi}\otimes\rho$ where $\ket{\phi}$ is with A and $\rho$ is jointly with A and B and independent of $\ket{\phi}$. 
ii) Swap the secret $\ket{\phi}$ with the arbitrary state $\ket{\psi}$ iii) Encode $\ket{\psi}$ but using $\mathcal{R}_A^\dagger\otimes\mathcal{I}_B$ by acting on the state $\ket{\psi}\bra{\psi}\otimes\rho$. Note that all these steps do not involve any operations on the shares in $B$. After these steps, the final state of qudits with $A$ and $B$ is the same as $\mathcal{E}\ket{\psi}\ket{0}$. The recovery operation by any authorized set from the $n$ shares remains the same as before but the state recovered is $\ket{\psi}$.
\end{proof}

Application of Lemma~\ref{lm:sec-rep} in the proof of our next lemma is similar to \cite[Theorem~6]{gottesman00}. However, Lemma~\ref{lm:sec-rep} is convenient and sufficient for our work. In the next theorem, we prove a lower bound on the communication cost for a $((k,n,d))$ quantum secret sharing scheme. 
We build on the ideas of Gottesman \cite{gottesman00} and Huang et al \cite{huang16}.
\begin{lemma}
\label{lm:bound-helps}
In any $((k,2k-1,d))$ QSS scheme, which recovers a secret of dimension $M$ from any set of $d$ shares, the total communication to the combiner from any $d-k+1$ shares among the $d$ shares is of dimension at least $M$.
\end{lemma}

\begin{proof}
We prove this by means of a communication protocol between Alice and Bob based on the QSS scheme. 
Alice needs to send an arbitrary state $\ket{\psi}$ of dimension $M$ to Bob.

First, encode the pure state $\ket{0}$ using the given QSS scheme. Consider any set of $d$ participants $D$ such that each participant in $D$ can send a part of its share to the combiner to recover the secret.  
Consider any subset $L\subseteq D$ with $d-k+1$ shares.

A third party, say Carol, is given the $k-1$ shares from the set $D\backslash L$. Alice is given the $d-k+1$ shares from $L$ and all the remaining $2k-1-d$ shares in the scheme. If Bob wants to reconstruct the secret by accessing some qudits from each of the $d$ shares in $D$, both Alice and Carol have to communicate some qudits from each share in $L$ and $D\backslash L$ respectively. Next, Carol sends the qudits needed for this reconstruction from each share in $D\backslash L$ to Bob.

Clearly, Bob has no prior information on $\ket{\psi}$ even though he may share some entanglement with Alice due to qudits he received earlier from Carol. Now, instead of directly transmitting $\ket{\psi}$ to Bob, Alice can exploit the secret sharing scheme for the communication. Using the authorized set of $k$ shares she already has, Alice replaces the secret $\ket{0}$ in the scheme with $\ket{\psi}$ (by Lemma~\ref{lm:sec-rep}). Then, she transmits the qudits from the shares in $L$ which Bob needs to reconstruct the encoded secret. Now, Bob uses the qudits received from shares in both $L$ and $D\backslash L$ to reconstruct the secret $\ket{\psi}$. By Lemma~\ref{lm:lim-qc}, the communication from the 
shares in $L$ has to be at least $M$.
\end{proof}

\begin{theorem}[Lower bound on communication cost]\label{th:co-lbounds} 
In any $((k,2k-1,d))$ quantum secret sharing scheme, recovery of a secret of dimension $M$ from $d$ shares requires communication of a state from a Hilbert space of dimension at least $M^{d/(d-k+1)}$ to the combiner.
\end{theorem}

\begin{proof}

Consider any set of $d$ participants $D$ such that each participant in $D$ can send a part of its share to the combiner to recover the secret. Label the part of $i$th share in $D$ communicated to the combiner as $H_i$ such that 
\begin{eqnarray}
\dim(H_1)\geq\dim(H_2)\geq\hdots\geq\dim(H_d)\label{eq:share-sizes}
\end{eqnarray}

Applying Lemma \ref{lm:bound-helps} for the set $\{H_k, H_{k+1},\hdots H_d\}$ which is the overall communication from a set of $d-k+1$ shares, 
\begin{eqnarray}
\prod_{i=k}^d \dim(H_i) \geq M \label{eq:bound_some_helps}
\end{eqnarray}
Then by Eq.~\eqref{eq:share-sizes}, we have 
\begin{eqnarray}
\dim(H_{k}) &\geq M^{1/(d-k+1)} \mbox{ and }
\dim(H_{i}) &\geq M^{1/(d-k+1)} \label{eq:bound_one_help}
\end{eqnarray}
for $1\leq i\leq k$. 
From Eq.~\eqref{eq:bound_some_helps}~and~\eqref{eq:bound_one_help}, the communication to the combiner from the $d$ shares in $D$ can be lower bounded as
\begin{eqnarray}
\prod_{i=1}^d \dim (H_i) &=& \prod_{i=1}^{k-1} \dim (H_i)\prod_{i=k}^d\dim (H_i)\\
&\geq& (\prod_{i=1}^{k-1}M^{1/(d-k+1)}) M 
= M^{d/d-k+1)}
\end{eqnarray}
This shows that the set of $d$ participants in $D$ must communicate a state that is in a Hilbert space of dimension atleast $M^{d/(d-k+1)}$. This completes the proof.
\end{proof}

If we let $M=q^{\ell}$, then we obtain the following corollary which immediately implies the optimality of the proposed schemes. 
\begin{corollary}[Optimality of proposed schemes]
Any $((k,2k-1,d))$ QSS  scheme sharing $\ell$ qudits incurs a communication cost of $\geq \frac{d\ell}{d-k+1}$ qudits.
The $((k,2k-1,d))$ QSS scheme of Theorem~\ref{th:ce-qss} has optimal communication cost 
(for fixed $d$).
\end{corollary}

In this paper we have proposed communication efficient quantum secret sharing schemes and demonstrated their  optimality with respect to communication cost. 
There are many further directions for research, some  which generalize the classical analogues \cite{bitar16,huang16,huang17,wang08,penas18,yan18} to the quantum setting. 
For instance, it is natural to study  secret sharing schemes that are efficient with variable $d$ as studied classically in \cite{bitar16}. 
Another direction for research is that of general access structures.


\begin{thebibliography}{23}%
\makeatletter
\providecommand \@ifxundefined [1]{%
 \@ifx{#1\undefined}
}%
\providecommand \@ifnum [1]{%
 \ifnum #1\expandafter \@firstoftwo
 \else \expandafter \@secondoftwo
 \fi
}%
\providecommand \@ifx [1]{%
 \ifx #1\expandafter \@firstoftwo
 \else \expandafter \@secondoftwo
 \fi
}%
\providecommand \natexlab [1]{#1}%
\providecommand \enquote  [1]{``#1''}%
\providecommand \bibnamefont  [1]{#1}%
\providecommand \bibfnamefont [1]{#1}%
\providecommand \citenamefont [1]{#1}%
\providecommand \href@noop [0]{\@secondoftwo}%
\providecommand \href [0]{\begingroup \@sanitize@url \@href}%
\providecommand \@href[1]{\@@startlink{#1}\@@href}%
\providecommand \@@href[1]{\endgroup#1\@@endlink}%
\providecommand \@sanitize@url [0]{\catcode `\\12\catcode `\$12\catcode
  `\&12\catcode `\#12\catcode `\^12\catcode `\_12\catcode `\%12\relax}%
\providecommand \@@startlink[1]{}%
\providecommand \@@endlink[0]{}%
\providecommand \url  [0]{\begingroup\@sanitize@url \@url }%
\providecommand \@url [1]{\endgroup\@href {#1}{\urlprefix }}%
\providecommand \urlprefix  [0]{URL }%
\providecommand \Eprint [0]{\href }%
\providecommand \doibase [0]{http://dx.doi.org/}%
\providecommand \selectlanguage [0]{\@gobble}%
\providecommand \bibinfo  [0]{\@secondoftwo}%
\providecommand \bibfield  [0]{\@secondoftwo}%
\providecommand \translation [1]{[#1]}%
\providecommand \BibitemOpen [0]{}%
\providecommand \bibitemStop [0]{}%
\providecommand \bibitemNoStop [0]{.\EOS\space}%
\providecommand \EOS [0]{\spacefactor3000\relax}%
\providecommand \BibitemShut  [1]{\csname bibitem#1\endcsname}%
\let\auto@bib@innerbib\@empty
\bibitem [{\citenamefont {Hillery}\ \emph {et~al.}(1999)\citenamefont
  {Hillery}, \citenamefont {Buzek},\ and\ \citenamefont
  {Berthiaume}}]{hillery99}%
  \BibitemOpen
  \bibfield  {author} {\bibinfo {author} {\bibfnamefont {M.}~\bibnamefont
  {Hillery}}, \bibinfo {author} {\bibfnamefont {V.}~\bibnamefont {Buzek}}, \
  and\ \bibinfo {author} {\bibfnamefont {A.}~\bibnamefont {Berthiaume}},\
  }\href {https://journals.aps.org/pra/abstract/10.1103/PhysRevA.59.1829}
  {\bibfield  {journal} {\bibinfo  {journal} {Phys. Rev. A}\ }\textbf {\bibinfo
  {volume} {59}},\ \bibinfo {pages} {1829} (\bibinfo {year}
  {1999})}\BibitemShut {NoStop}%
\bibitem [{\citenamefont {Karlsson}\ \emph {et~al.}(1999)\citenamefont
  {Karlsson}, \citenamefont {Koashi},\ and\ \citenamefont
  {Imoto}}]{karlsson99}%
  \BibitemOpen
  \bibfield  {author} {\bibinfo {author} {\bibfnamefont {A.}~\bibnamefont
  {Karlsson}}, \bibinfo {author} {\bibfnamefont {M.}~\bibnamefont {Koashi}}, \
  and\ \bibinfo {author} {\bibfnamefont {N.}~\bibnamefont {Imoto}},\ }\href
  {https://journals.aps.org/pra/abstract/10.1103/PhysRevA.59.162} {\bibfield
  {journal} {\bibinfo  {journal} {Phys. Rev. A}\ }\textbf {\bibinfo {volume}
  {59}},\ \bibinfo {pages} {162} (\bibinfo {year} {1999})}\BibitemShut
  {NoStop}%
\bibitem [{\citenamefont {Cleve}\ \emph {et~al.}(1999)\citenamefont {Cleve},
  \citenamefont {Gottesman},\ and\ \citenamefont {Lo}}]{cleve99}%
  \BibitemOpen
  \bibfield  {author} {\bibinfo {author} {\bibfnamefont {R.}~\bibnamefont
  {Cleve}}, \bibinfo {author} {\bibfnamefont {D.}~\bibnamefont {Gottesman}}, \
  and\ \bibinfo {author} {\bibfnamefont {H.-K.}\ \bibnamefont {Lo}},\ }\href
  {https://journals.aps.org/prl/abstract/10.1103/PhysRevLett.83.648} {\bibfield
   {journal} {\bibinfo  {journal} {Phys. Rev. Lett.}\ }\textbf {\bibinfo
  {volume} {83}},\ \bibinfo {pages} {648} (\bibinfo {year} {1999})}\BibitemShut
  {NoStop}%
\bibitem [{\citenamefont {Gottesman}(2000)}]{gottesman00}%
  \BibitemOpen
  \bibfield  {author} {\bibinfo {author} {\bibfnamefont {D.}~\bibnamefont
  {Gottesman}},\ }\href {\doibase 10.1103/PhysRevA.61.042311} {\bibfield
  {journal} {\bibinfo  {journal} {Phys. Rev. A}\ }\textbf {\bibinfo {volume}
  {61}},\ \bibinfo {pages} {042311} (\bibinfo {year} {2000})}\BibitemShut
  {NoStop}%
\bibitem [{\citenamefont {Smith}(2000)}]{smith99}%
  \BibitemOpen
  \bibfield  {author} {\bibinfo {author} {\bibfnamefont {A.~D.}\ \bibnamefont
  {Smith}},\ }\href {https://arxiv.org/abs/quant-ph/0001087} {\bibfield
  {journal} {\bibinfo  {journal} {e-print quant-ph/0001087}\ } (\bibinfo {year}
  {2000})}\BibitemShut {NoStop}%
\bibitem [{\citenamefont {Ogawa}\ \emph {et~al.}(2005)\citenamefont {Ogawa},
  \citenamefont {Sasaki}, \citenamefont {Iwamoto},\ and\ \citenamefont
  {Yamamoto}}]{ogawa05}%
  \BibitemOpen
  \bibfield  {author} {\bibinfo {author} {\bibfnamefont {T.}~\bibnamefont
  {Ogawa}}, \bibinfo {author} {\bibfnamefont {A.}~\bibnamefont {Sasaki}},
  \bibinfo {author} {\bibfnamefont {M.}~\bibnamefont {Iwamoto}}, \ and\
  \bibinfo {author} {\bibfnamefont {H.}~\bibnamefont {Yamamoto}},\ }\href
  {https://journals.aps.org/pra/abstract/10.1103/PhysRevA.72.032318} {\bibfield
   {journal} {\bibinfo  {journal} {Phys. rev. A}\ }\textbf {\bibinfo {volume}
  {72}},\ \bibinfo {pages} {032318} (\bibinfo {year} {2005})}\BibitemShut
  {NoStop}%
\bibitem [{\citenamefont {Markham}\ and\ \citenamefont
  {Sanders}(2008)}]{markham08}%
  \BibitemOpen
  \bibfield  {author} {\bibinfo {author} {\bibfnamefont {D.}~\bibnamefont
  {Markham}}\ and\ \bibinfo {author} {\bibfnamefont {B.~C.}\ \bibnamefont
  {Sanders}},\ }\href
  {https://journals.aps.org/pra/abstract/10.1103/PhysRevA.78.042309} {\bibfield
   {journal} {\bibinfo  {journal} {Phys. Rev. A}\ }\textbf {\bibinfo {volume}
  {78}},\ \bibinfo {pages} {042309} (\bibinfo {year} {2008})}\BibitemShut
  {NoStop}%
\bibitem [{\citenamefont {Zhang}\ and\ \citenamefont
  {Matsumoto}(2015)}]{zhang15}%
  \BibitemOpen
  \bibfield  {author} {\bibinfo {author} {\bibfnamefont {P.}~\bibnamefont
  {Zhang}}\ and\ \bibinfo {author} {\bibfnamefont {R.}~\bibnamefont
  {Matsumoto}},\ }\href
  {https://link.springer.com/article/10.1007/s11128-014-0863-2} {\bibfield
  {journal} {\bibinfo  {journal} {Quantum Information Processing}\ }\textbf
  {\bibinfo {volume} {14}},\ \bibinfo {pages} {715} (\bibinfo {year}
  {2015})}\BibitemShut {NoStop}%
\bibitem [{\citenamefont {Bitar}\ and\ \citenamefont
  {El~Rouayheb}(2016)}]{bitar16}%
  \BibitemOpen
  \bibfield  {author} {\bibinfo {author} {\bibfnamefont {R.}~\bibnamefont
  {Bitar}}\ and\ \bibinfo {author} {\bibfnamefont {S.}~\bibnamefont
  {El~Rouayheb}},\ }in\ \href {https://ieeexplore.ieee.org/document/7541528}
  {\emph {\bibinfo {booktitle} {Proc. 2016 IEEE Intl. Symposium on Information
  Theory, Barcelona, Spain}}}\ (\bibinfo {year} {2016})\ pp.\ \bibinfo {pages}
  {1396--1400},\ \bibinfo {note} {extended version,
  \href{https://arxiv.org/abs/1512.02990}{arXiv:1512.02990}}\BibitemShut
  {NoStop}%
\bibitem [{\citenamefont {Huang}\ \emph {et~al.}(2016)\citenamefont {Huang},
  \citenamefont {Langberg}, \citenamefont {Kliewet},\ and\ \citenamefont
  {Bruck}}]{huang16}%
  \BibitemOpen
  \bibfield  {author} {\bibinfo {author} {\bibfnamefont {W.}~\bibnamefont
  {Huang}}, \bibinfo {author} {\bibfnamefont {M.}~\bibnamefont {Langberg}},
  \bibinfo {author} {\bibfnamefont {J.}~\bibnamefont {Kliewet}}, \ and\
  \bibinfo {author} {\bibfnamefont {J.}~\bibnamefont {Bruck}},\ }\href
  {https://ieeexplore.ieee.org/abstract/document/7587343} {\bibfield  {journal}
  {\bibinfo  {journal} {IEEE Trans. Inform. Theory}\ }\textbf {\bibinfo
  {volume} {62}},\ \bibinfo {pages} {7195 } (\bibinfo {year}
  {2016})}\BibitemShut {NoStop}%
\bibitem [{\citenamefont {Huang}\ and\ \citenamefont {Bruck}(2017)}]{huang17}%
  \BibitemOpen
  \bibfield  {author} {\bibinfo {author} {\bibfnamefont {W.}~\bibnamefont
  {Huang}}\ and\ \bibinfo {author} {\bibfnamefont {J.}~\bibnamefont {Bruck}},\
  }in\ \href {https://ieeexplore.ieee.org/abstract/document/8006842} {\emph
  {\bibinfo {booktitle} {Proc. 2017 IEEE Intl. Symposium on Information Theory,
  Aachen, Germany}}}\ (\bibinfo {year} {2017})\ pp.\ \bibinfo {pages}
  {1813--1817}\BibitemShut {NoStop}%
\bibitem [{\citenamefont {Wang}\ and\ \citenamefont {Wong}(2008)}]{wang08}%
  \BibitemOpen
  \bibfield  {author} {\bibinfo {author} {\bibfnamefont {H.}~\bibnamefont
  {Wang}}\ and\ \bibinfo {author} {\bibfnamefont {D.~S.}\ \bibnamefont
  {Wong}},\ }\href {https://ieeexplore.ieee.org/abstract/document/4418504}
  {\bibfield  {journal} {\bibinfo  {journal} {IEEE Trans. Inform. Theory}\
  }\textbf {\bibinfo {volume} {54}},\ \bibinfo {pages} {473} (\bibinfo {year}
  {2008})}\BibitemShut {NoStop}%
\bibitem [{\citenamefont {Marti\'nez-Pe\~nas}(2018)}]{penas18}%
  \BibitemOpen
  \bibfield  {author} {\bibinfo {author} {\bibfnamefont {U.}~\bibnamefont
  {Marti\'nez-Pe\~nas}},\ }\href
  {https://ieeexplore.ieee.org/abstract/document/8331930} {\bibfield  {journal}
  {\bibinfo  {journal} {IEEE Trans. Inform. Theory}\ }\textbf {\bibinfo
  {volume} {64}},\ \bibinfo {pages} {4191 } (\bibinfo {year}
  {2018})}\BibitemShut {NoStop}%
\bibitem [{\citenamefont {Yan}\ \emph {et~al.}(2018)\citenamefont {Yan},
  \citenamefont {Lin}, \citenamefont {Lu},\ and\ \citenamefont {Tang}}]{yan18}%
  \BibitemOpen
  \bibfield  {author} {\bibinfo {author} {\bibfnamefont {X.}~\bibnamefont
  {Yan}}, \bibinfo {author} {\bibfnamefont {C.}~\bibnamefont {Lin}}, \bibinfo
  {author} {\bibfnamefont {R.}~\bibnamefont {Lu}}, \ and\ \bibinfo {author}
  {\bibfnamefont {C.}~\bibnamefont {Tang}},\ }\href
  {https://ieeexplore.ieee.org/abstract/document/8369107} {\bibfield  {journal}
  {\bibinfo  {journal} {IEEE Communications Letters}\ }\textbf {\bibinfo
  {volume} {22}},\ \bibinfo {pages} {1556 } (\bibinfo {year}
  {2018})}\BibitemShut {NoStop}%
\bibitem [{\citenamefont {Tittel}\ \emph {et~al.}(2001)\citenamefont {Tittel},
  \citenamefont {Zbinden},\ and\ \citenamefont {Gisin}}]{tittel01}%
  \BibitemOpen
  \bibfield  {author} {\bibinfo {author} {\bibfnamefont {W.}~\bibnamefont
  {Tittel}}, \bibinfo {author} {\bibfnamefont {H.}~\bibnamefont {Zbinden}}, \
  and\ \bibinfo {author} {\bibfnamefont {N.}~\bibnamefont {Gisin}},\ }\href
  {https://journals.aps.org/pra/abstract/10.1103/PhysRevA.63.042301} {\bibfield
   {journal} {\bibinfo  {journal} {Phys. Rev. A}\ }\textbf {\bibinfo {volume}
  {63}},\ \bibinfo {pages} {042301} (\bibinfo {year} {2001})}\BibitemShut
  {NoStop}%
\bibitem [{\citenamefont {Imai}\ \emph {et~al.}(2003)\citenamefont {Imai},
  \citenamefont {MÃŒller-Quade}, \citenamefont {Nascimento}, \citenamefont
  {Tuyls},\ and\ \citenamefont {Winter}}]{imai03}%
  \BibitemOpen
  \bibfield  {author} {\bibinfo {author} {\bibfnamefont {H.}~\bibnamefont
  {Imai}}, \bibinfo {author} {\bibfnamefont {J.}~\bibnamefont {MÃŒller-Quade}},
  \bibinfo {author} {\bibfnamefont {A.~C.}\ \bibnamefont {Nascimento}},
  \bibinfo {author} {\bibfnamefont {P.}~\bibnamefont {Tuyls}}, \ and\ \bibinfo
  {author} {\bibfnamefont {A.}~\bibnamefont {Winter}},\ }\href
  {https://arxiv.org/abs/quant-ph/0311136} {\bibfield  {journal} {\bibinfo
  {journal} {e-print quant-ph/0311136}\ } (\bibinfo {year} {2003})}\BibitemShut
  {NoStop}%
\bibitem [{\citenamefont {Wei}\ \emph {et~al.}(2013)\citenamefont {Wei},
  \citenamefont {Ma},\ and\ \citenamefont {Yang}}]{wei13}%
  \BibitemOpen
  \bibfield  {author} {\bibinfo {author} {\bibfnamefont {K.~J.}\ \bibnamefont
  {Wei}}, \bibinfo {author} {\bibfnamefont {H.~Q.}\ \bibnamefont {Ma}}, \ and\
  \bibinfo {author} {\bibfnamefont {J.~H.}\ \bibnamefont {Yang}},\ }\href
  {https://www.osapublishing.org/oe/abstract.cfm?uri=oe-21-14-16663} {\bibfield
   {journal} {\bibinfo  {journal} {Optics express}\ }\textbf {\bibinfo {volume}
  {21}},\ \bibinfo {pages} {16663 } (\bibinfo {year} {2013})}\BibitemShut
  {NoStop}%
\bibitem [{\citenamefont {Hao}\ \emph {et~al.}(2011)\citenamefont {Hao},
  \citenamefont {Wang},\ and\ \citenamefont {Long}}]{hao11}%
  \BibitemOpen
  \bibfield  {author} {\bibinfo {author} {\bibfnamefont {L.}~\bibnamefont
  {Hao}}, \bibinfo {author} {\bibfnamefont {C.}~\bibnamefont {Wang}}, \ and\
  \bibinfo {author} {\bibfnamefont {G.~L.}\ \bibnamefont {Long}},\ }\href
  {https://www.sciencedirect.com/science/article/pii/S0030401811003105}
  {\bibfield  {journal} {\bibinfo  {journal} {Optics Communications}\ }\textbf
  {\bibinfo {volume} {284}},\ \bibinfo {pages} {3639 } (\bibinfo {year}
  {2011})}\BibitemShut {NoStop}%
\bibitem [{\citenamefont {Bogdanski}\ \emph {et~al.}(2008)\citenamefont
  {Bogdanski}, \citenamefont {Rafiei},\ and\ \citenamefont
  {Bourennane}}]{bogdanski08}%
  \BibitemOpen
  \bibfield  {author} {\bibinfo {author} {\bibfnamefont {J.}~\bibnamefont
  {Bogdanski}}, \bibinfo {author} {\bibfnamefont {N.}~\bibnamefont {Rafiei}}, \
  and\ \bibinfo {author} {\bibfnamefont {M.}~\bibnamefont {Bourennane}},\
  }\href {https://journals.aps.org/pra/abstract/10.1103/PhysRevA.78.062307}
  {\bibfield  {journal} {\bibinfo  {journal} {Phys. Rev. A}\ }\textbf {\bibinfo
  {volume} {78}},\ \bibinfo {pages} {062307} (\bibinfo {year}
  {2008})}\BibitemShut {NoStop}%
\bibitem [{\citenamefont {Bell}\ \emph {et~al.}(2014)\citenamefont {Bell},
  \citenamefont {Markham}, \citenamefont {Herrera-MartÃ­}, \citenamefont
  {Marin}, \citenamefont {Wadsworth}, \citenamefont {Rarity},\ and\
  \citenamefont {Tame}}]{bell14}%
  \BibitemOpen
  \bibfield  {author} {\bibinfo {author} {\bibfnamefont {B.~A.}\ \bibnamefont
  {Bell}}, \bibinfo {author} {\bibfnamefont {D.}~\bibnamefont {Markham}},
  \bibinfo {author} {\bibfnamefont {D.~A.}\ \bibnamefont {Herrera-MartÃ­}},
  \bibinfo {author} {\bibfnamefont {A.}~\bibnamefont {Marin}}, \bibinfo
  {author} {\bibfnamefont {W.~J.}\ \bibnamefont {Wadsworth}}, \bibinfo {author}
  {\bibfnamefont {J.~G.}\ \bibnamefont {Rarity}}, \ and\ \bibinfo {author}
  {\bibfnamefont {M.~S.}\ \bibnamefont {Tame}},\ }\href
  {https://www.nature.com/articles/ncomms6480} {\bibfield  {journal} {\bibinfo
  {journal} {Nature communications}\ }\textbf {\bibinfo {volume} {5}},\
  \bibinfo {pages} {5480} (\bibinfo {year} {2014})}\BibitemShut {NoStop}%
\bibitem [{\citenamefont {Schmid}\ \emph {et~al.}(2005)\citenamefont {Schmid},
  \citenamefont {Trojek}, \citenamefont {Bourennane}, \citenamefont
  {Kurtsiefer}, \citenamefont {\.Zukowski},\ and\ \citenamefont
  {Weinfurter}}]{schmid05}%
  \BibitemOpen
  \bibfield  {author} {\bibinfo {author} {\bibfnamefont {C.}~\bibnamefont
  {Schmid}}, \bibinfo {author} {\bibfnamefont {P.}~\bibnamefont {Trojek}},
  \bibinfo {author} {\bibfnamefont {M.}~\bibnamefont {Bourennane}}, \bibinfo
  {author} {\bibfnamefont {C.}~\bibnamefont {Kurtsiefer}}, \bibinfo {author}
  {\bibfnamefont {M.}~\bibnamefont {\.Zukowski}}, \ and\ \bibinfo {author}
  {\bibfnamefont {H.}~\bibnamefont {Weinfurter}},\ }\href
  {https://journals.aps.org/prl/abstract/10.1103/PhysRevLett.95.230505}
  {\bibfield  {journal} {\bibinfo  {journal} {Phys. Rev. Lett.}\ }\textbf
  {\bibinfo {volume} {95}},\ \bibinfo {pages} {230505} (\bibinfo {year}
  {2005})}\BibitemShut {NoStop}%
\bibitem [{\citenamefont {Gaertner}\ \emph {et~al.}(2007)\citenamefont
  {Gaertner}, \citenamefont {Kurtsiefer}, \citenamefont {Bourennane},\ and\
  \citenamefont {Weinfurter}}]{gaertner07}%
  \BibitemOpen
  \bibfield  {author} {\bibinfo {author} {\bibfnamefont {S.}~\bibnamefont
  {Gaertner}}, \bibinfo {author} {\bibfnamefont {C.}~\bibnamefont
  {Kurtsiefer}}, \bibinfo {author} {\bibfnamefont {M.}~\bibnamefont
  {Bourennane}}, \ and\ \bibinfo {author} {\bibfnamefont {H.}~\bibnamefont
  {Weinfurter}},\ }\href
  {https://journals.aps.org/prl/abstract/10.1103/PhysRevLett.98.020503}
  {\bibfield  {journal} {\bibinfo  {journal} {Phys. Rev. Lett.}\ }\textbf
  {\bibinfo {volume} {98}},\ \bibinfo {pages} {020503} (\bibinfo {year}
  {2007})}\BibitemShut {NoStop}%
\bibitem [{\citenamefont {Fortescue}\ and\ \citenamefont {Gour}(2012)}]{ben12}%
  \BibitemOpen
  \bibfield  {author} {\bibinfo {author} {\bibfnamefont {B.}~\bibnamefont
  {Fortescue}}\ and\ \bibinfo {author} {\bibfnamefont {G.}~\bibnamefont
  {Gour}},\ }\href {https://ieeexplore.ieee.org/abstract/document/6225432}
  {\bibfield  {journal} {\bibinfo  {journal} {IEEE Trans. Inform. Theory}\
  }\textbf {\bibinfo {volume} {58}},\ \bibinfo {pages} {6659 } (\bibinfo {year}
  {2012})}\BibitemShut {NoStop}%
\end{thebibliography}
\end{document}